\documentclass[12pt]{article}

\usepackage{scicite}
\usepackage{amsmath,amsthm,amssymb,amsfonts}

\newtheorem{theorem}{Theorem}
\newtheorem{lemma}{Lemma}
\newtheorem{corollary}{Corollary}
\newtheorem{remark}{Remark}

\usepackage{times}

\newenvironment{sciabstract}{%
\begin{quote} \bf}
{\end{quote}}

\newcounter{lastnote}

\title{A Note on the Convexity of $\log \det ( I  + KX^{-1} )$ and its Constrained Optimization Representation}

\author
{Kwang-Ki K. Kim$^{1\ast}$ \\
\normalsize{$^{1}$School of Electrical and Computer Engineering, Georgia Institute of Technology}
\\
\\
\normalsize{$^\ast$Correspondence: {\tt kwangki.kim@ece.gatech.edu}}
}

\date{}

\begin{document} 


\baselineskip24pt


\maketitle


\begin{sciabstract} \rm
This note provides another proof for the {\em convexity} ({\em strict convexity}) of $\log \det ( I  + KX^{-1} )$
over the positive definite cone for any given positive semidefinite matrix $K \succeq 0$ (positive definite matrix $K \succ 0$)
and the {\em strictly convexity} of $\log \det (K + X^{-1})$ over the positive definite cone for any given $K \succeq 0$.
Equivalent optimization representation with linear matrix inequalities (LMIs) for the functions $\log \det ( I  + KX^{-1} )$ and $\log \det (K + X^{-1})$ are presented.
Their optimization representations with LMI constraints can be particularly useful for some related synthetic design problems.
\end{sciabstract}



\section*{Results}

\begin{theorem}\label{theorem:convexity}
For given $K \succeq 0$ the function
\[
f(X) = \log \det ( I  + KX^{-1} )
\]
is convex over the positive semidefinite cone $\mathbb S_{++}^{n} = \{ X \in \mathbb{C}^{n \times n} : X = X^* \succ 0 \}$
and strictly convex if $K \succ 0$.
\end{theorem}

The first proof for Theorem~\ref{theorem:convexity} was presented by \cite{DiggaviCover} and some other approaches to its proof were given in~\cite{Mao,Kashyap,KimKim}.
The proof in~\cite{DiggaviCover} is an information theoretic approach, \cite{Mao} gives a complicated proof to correct the incomplete proof in \cite{Kashyap}, and
\cite{KimKim} provides a simple proof based on the theory of spectral functions.
We provide a different approach to show (strict) convexity of the above function $f : S_{++}^{n} \rightarrow \mathbb R$
for which the convexity of a companion function $\log \det ( X^{-1} )$ is exploited.\footnote{The convexity of $\log \det ( X^{-1} )$ is not trivial, but nevertheless well known and can be found in many convex analysis textbooks.}
As the convexity of $\log \det ( X^{-1} )$ itself is not trivial,
our proof is not stand-alone and the purpose of this note is not to claim to fame but to draw attentions to some equivalent convex programs with linear matrix inequalities for $f$.

\begin{lemma}\label{lemma:sdp}
For $K \succeq 0$, the value of function $f(X)$ is the same to the optimal value of a semidefinite program for every $X \succ 0$:
\[
\begin{split}
\log \det ( I  + KX^{-1} )  \equiv 
& \min_{Z \succ 0}\, \log \det (Z^{-1}) \\
& \ \mbox{\rm s.t.} \ \left[ \begin{array}{@{}c@{\,\,}c@{}} I - Z & K^{1/2} \\ K^{1/2} & X + K \end{array} \right] \succeq 0 \,.
\end{split}
\]
\end{lemma}
\begin{proof}[Proof of Lemma~\ref{lemma:sdp}]
From the Sylvester's determinant theorem,
\[
\log \det ( I  + KX^{-1} ) = \log \det ( I  + K^{1/2} X^{-1} K^{1/2} ) \,.
\]
Introducing a slack variable $Z \succ 0$ satisfying the inequality $(I  + K^{1/2} X^{-1} K^{1/2}) \preceq Z^{-1}$, we have the following equivalent representations:
\[
\begin{split}
(I  + K^{1/2} X^{-1} K^{1/2})^{-1} \succeq Z 
& \ \Leftrightarrow \ I - K^{1/2} ( X + K)^{-1} K^{1/2} \succeq Z \\
& \ \Leftrightarrow \ \left[ \begin{array}{@{}c@{\,\,}c@{}} I - Z & K^{1/2} \\ K^{1/2} & X + K \end{array} \right] \succeq 0
\end{split}
\]
where the last equivalence follows from the Schur complement.
Due to the monotonicity of $\log \det (\cdot)$, the minimum $Z \succ 0$ satisfying the equality gives the value of $\log \det ( I  + KX^{-1} )$ for all $X$.
Since the constraint set is compact for all $X \succ 0$ and $K \succeq 0$
and $- \log \det: \mathbb S_{++}^{n} \rightarrow \mathbb R$ is strictly convex,
the minimum denoted by $Z^{\star}(X)$ always exists and is unique for every $X \succ 0$.
\end{proof}

\begin{proof}[Proof of Theorem~\ref{theorem:convexity}]
From Lemma~\ref{lemma:sdp} and linearity of the constraint, for all $\lambda \in [0,1]$ and $X,Y \succ 0$
\[
\begin{split}
f( \lambda X + (1- \lambda) Y) 
& = 
\min_{Z \succ 0}\ \log \det (Z^{-1}) \\
& \ \ \ \ \ \ \   \mbox{\rm s.t.} \ \,  \lambda \left[ \begin{array}{@{}c@{\,\,}c@{}} I - Z & K^{1/2} \\ K^{1/2} & X + K \end{array} \right]  
+ (1-\lambda) \left[ \begin{array}{@{}c@{\,\,}c@{}} I - Z & K^{1/2} \\ K^{1/2} & Y + K \end{array} \right]  \succeq 0 \\
& \leq
 - \log \det \left( \lambda Z^{\star}(X) + (1-\lambda) Z^{\star}(Y) \right) \\
 & \leq
  - \lambda \log \det (Z^{\star}(X)) - (1-\lambda) \log \det (Z^{\star}(Y))  \\
  & = 
  \lambda f(X) + (1-\lambda) f(Y)
\end{split}
\]
for which the second inequality is due to the (strict) convexity of $- \log \det (\cdot)$.
Since $- \log \det: \mathbb S_{++}^{n} \rightarrow \mathbb R$ is strictly convex, it is straightforward that
$\log \det ( I  + KX^{-1} )$ is {\em strictly convex} if $X \neq Y$ implies $Z^{\star}(X) \neq Z^{\star}(Y)$, which is equivalent to the condition $K \succ 0$.
\end{proof}

\begin{corollary}\label{corollary:sdp:strictconvex}
For given $K \succeq 0$ the function 
\[
g(X) = \log \det ( K + X^{-1} )
\]
is {\em strictly convex} over $\mathbb S_{++}^{n}$.
\end{corollary}

Similar to Lemma~\ref{lemma:sdp}, for $K \succeq 0$ the above function $g(X)$ has the following convex program with linear matrix inequalities:
\[
\begin{split}
\log \det ( K + X^{-1} )  \equiv
& \min_{Z \succ 0}\, \log \det (Z^{-1}) \\
& \ \mbox{\rm s.t.} \ \left[ \begin{array}{@{}c@{\,\,}c@{}} X - Z & X K^{1/2} \\ K^{1/2} X & I + K^{1/2} X K^{1/2} \end{array} \right] \succeq 0 \,.
\end{split}
\]

\section*{The Use of a MaxDet-LMI Representation}
The convex optimization problems with linear matrix inequalities for the functions $f(X) = \log \det ( I  + KX^{-1} )$ and $g(X) = \log \det ( K + X^{-1} )$
can be particularly useful for some optimization related to the variable $X$.
For example, consider a constrained optimization of the form\footnote{This form of optimization is presented in~\cite{Tanaka:2014:TAC} for which the corresponding optimization is to compute the sequential rate distortion function for a stationary Gauss-Markov process with a linear sensing function.}
\[
\begin{split}
& \min_{X \succ 0}\, g(X) \\
& \ \, \mbox{\rm s.t.} \ H (X) \succeq 0 
\end{split}
\]
where $H : \mathbb S_{++}^{n} \rightarrow \mathbb S^{m}$ is a matrix-valued {\em affine} function.
Then this optimization can be equivalently rewritten by negative logarithmic determinant minimization over linear matrix inequalities:
\[
\begin{split}
 \min_{X \succ 0, \, Z\succ 0} & \, -\log \det (Z) \\
 \mbox{\rm s.t.} \ \ \ \ & \ H (X) \succeq 0, \\
& 
\left[ \begin{array}{@{}c@{\,\,}c@{}} X - Z & X K^{1/2} \\ K^{1/2} X & I + K^{1/2} X K^{1/2} \end{array} \right] \succeq 0 \,,
\end{split}
\]
for which a slack variable $Z$ is introduced.
For the function $f(X) = \log \det ( I  + KX^{-1} )$, the same approach can be used.

\begin{remark}
In~\cite{vandenberghe1998determinant}, an overview of the applications of the determinant maximization problem with linear matrix inequalities to the computations of the Gaussian channel capacity is provided.
Similar to the results in~\cite{vandenberghe1998determinant}, additional structures on the matrix $X$ can be straightforwardly imposed in our minimization problems with linear matrix inequalities for the functions $f(X) = \log \det ( I  + KX^{-1} )$ and $g(X) = \log \det ( K + X^{-1} )$.
\end{remark}

\section*{Concluding Remarks}

In this note, we provide a matrix algebra approach to prove the convexity of $\log \det ( I + K X^{-1})$ and $\log \det ( K + X^{-1})$ for $K \succeq 0$
and the strict convexity counterpart. The proof does not stand alone, since it requires the convexity of $\log \det (X^{-1})$.
The method introducing a slack variable gives equivalent convex programs with linear matrix inequalities for those functions,
which can be used for constrained optimization with variable $X \succ 0$.


\bibliography{convexity-sdp}
\bibliographystyle{alpha}




\end{document}